\documentclass[11pt,twoside,letterpaper]{llncs}
\usepackage{times}
\usepackage[T1]{fontenc}   
\usepackage{amsmath}
\usepackage{amssymb}
\usepackage{fancybox}
\usepackage{graphics,graphicx}
\newtheorem{fact}{Fact}
\newcommand{\junk}[1]{}

\begin{document}
\title{Quantum and approximation  algorithms for maximum witnesses
of Boolean matrix products}

\author{
Miros{\l}aw Kowaluk
\inst{1}
\and
Andrzej Lingas
\inst{2}
\institute{Institute of Informatics, University of Warsaw, Warsaw, Poland. 
\texttt{kowaluk@mimuw.edu.pl}
\and
Department of Computer Science, Lund University, 22100 Lund, Sweden. 
\texttt{Andrzej.Lingas@cs.lth.se}}
}
\maketitle

\begin{abstract}
The problem of finding maximum (or minimum) witnesses of the Boolean
product of two Boolean matrices (MW for short) has a number of
important applications, in particular the all-pairs lowest common
ancestor (LCA) problem in directed acyclic graphs (dags).
The best known upper time-bound on the MW problem for $n\times n$
Boolean matrices of the form $O(n^{2.575})$ has not been substantially
improved since 2006.
In order to obtain faster algorithms for this problem, we study
quantum algorithms for MW and approximation algorithms for MW (in the
standard computational model).  Some of our quantum algorithms are
input or output sensitive.  Our fastest quantum algorithm for the MW
problem, and consequently for the related problems, runs in time
$\tilde{O}(n^{2+\lambda/2})=\tilde{O}(n^{2.434})$, where $\lambda $
satisfies the equation $\omega(1, \lambda, 1) = 1 + 1.5 \, \lambda $
and $\omega(1, \lambda, 1)$ is the exponent of the multiplication of
an $n \times n^{\lambda}$ matrix by an $n^{\lambda} \times n$ matrix.
Next, we consider a relaxed version of the MW problem (in the standard
model) asking for reporting a witness of bounded rank (the maximum
witness has rank 1) for each non-zero entry of the matrix product.
First, by adapting the fastest known algorithm for maximum
witnesses, we obtain an algorithm for the relaxed problem
that reports for each non-zero entry of the product matrix
a witness of rank at most $\ell$ in time
$\tilde{O}((n/\ell)n^{\omega(1,\log_n \ell,1)}).$ Then, by
reducing the relaxed problem to the so called $k$-witness problem, we
provide an algorithm that reports for each non-zero entry $C[i,j]$ of
the product matrix $C$ a witness of rank $O(\lceil W_C(i,j)/k\rceil
)$, where $W_C(i,j)$ is the number of witnesses for $C[i,j]$, with
high probability.  The algorithm runs in 
$\tilde{O}(n^{\omega}k^{0.4653} +n^{2+o(1)}k)$ time, where $\omega=\omega(1,1,1)$.
\end{abstract}
\section{Introduction}
If $A$ and $B$ are two $n\times n$ Boolean matrices
and $C$ is their Boolean matrix product then for any
entry $C[i,j]=1$ of $C,$ a  {\em witness} is an index
$k$ such that $A[i,k]\wedge B[k,j]=1.$
The largest  (or, smallest) possible witness for an entry is 
called the {\em maximum  witness} (or {\em minimum witness}, respectively)
for the entry.

The problem of finding ``witnesses'' of Boolean matrix product
has been studied for decades mostly because of its
applications to shortest-path problems \cite{AGM,AN96}.
The problem of finding maximum witnesses of Boolean matrix
product (MW for short) has been studied first  in \cite{CK07}  in order
to obtain  faster algorithms for all-pairs lowest common ancestor (LCA)
problem in directed acyclic graphs (dags) \cite{GIL}.
It has found many other applications since then 
including the all-pairs 
bottleneck weight path problem \cite{SYZ11}
and finding for a set of edges in a vertex-weighted graph
heaviest triangles including an edge from the set \cite{VWY10}.
The fastest known algorithm for the MW problem
and the aforementioned problems
 runs in $O(n^{2+\lambda})$ time \cite{CK07}, where $\lambda $
satisfies the equation $\omega(1, \lambda, 1) = 1 + 2\, \lambda $
and $\omega(1, \lambda, 1)$ is the exponent of the multiplication of
an $n \times n^{\lambda}$ matrix by an $n^{\lambda} \times n$
matrix. 
The currently best bounds on $\omega(1,\lambda,1)$
follow from a fact in \cite{C97,HP} (see 
Fact \ref{fact:1} in Preliminaries)
combined  with
the recent improved estimations
on the parameters $\omega=\omega(1,1,1)$ and $\alpha$ (see Preliminaries)
\cite{LG14,LGU}. They yield an
$O(n^{2.569})$ upper bound on the  running time of the algorithm 
(originally, $O(n^{2.575})$ \cite{CK07}).
For faster algorithms in sparse cases, see
\cite{CY}.

\junk{
Also in both areas, useful generalizations and/or specializations
of the problems of finding witnesses have been studied.
A natural generalization introduced for string matching in
\cite{M96} is to request up to $k$ witnesses
instead of a single one. It has been efficiently
solved  by using concepts from
group testing in \cite{ALL07} and conveyed
to Boolean matrix product in \cite{ALL07,GKL08}.
A natural specialization is to request minimum or maximum witnesses.
This specialization has been introduced and efficiently solved
in \cite{CK07} in the context of finding lowest common ancestors
in directed acyclic graphs and it found many other applications since
then (cf. \cite{CY,SYZ11,VWY10}).

In analogy to witnesses for Boolean matrix product,
if $a$ and $b$ are two $n$-dimensional Boolean vectors
and $c$ is their Boolean convolution then for any
coordinate $c_i=1$ of $c,$ a {\em witness} is an index
$l$ such that $a_l\wedge b_{i-l}=1$.
In contrast to string matching and Boolean matrix product,
the problem of computing the witnesses of Boolean vector
convolution does not seem to be explicitly studied
in the literature. On the other hand, Boolean vector
convolution is very much related to string matching \cite{FP74},
and hence the algorithms for reporting witness
or more generally up to $k$ witnesses can be easily conveyed
from stringology to Boolean vector convolution 
(see Proposition 1 in section 3).
}

In this paper, we study two different approaches to
deriving faster algorithms for the problem of
computing maximum (or minimum) witnesses of the Boolean product of two
$n\times n$ Boolean matrices (MW for short). The first approach
is to consider the MW problem in the more powerful 
model of quantum computation. The other approach
is to relax the MW problem (in the standard model)
by allowing its approximation.

In the first part of
our paper, we present quantum algorithms for the MW problem
assuming a Quantum Random Access Machine (QRAM) model
\cite{NC}.  First, we consider a straightforward algorithm
for MW that uses the quantum minimum search
due to D\"{u}rr and H{\o}yer \cite{DH96}
for each entry of the product matrix separately in
order to find its maximum witness
\footnote{For somewhat related applications of the quantum minimum search 
of D\"{u}rr and H{\o}yer to shortest path problems see \cite{NV}.}.
 It runs in $\tilde{O}(n^{2.5})$ time.
By adding as a preprocessing a known output-sensitive
quantum algorithm for Boolean matrix product,
 we obtain 
an output-sensitive quantum algorithm for MW
running in $\tilde{O}(n\sqrt s + s\sqrt n)$ time, where $s$ is the number
of non-zero  entries in the product matrix.
By refining the straightforward algorithm in a different way, we obtain also an input-sensitive
quantum algorithm for MW running in $\tilde{O}(n^2+n^{1.5}m^{0.5})$
time, where $m$ is the number of non-zero entries
in the sparsest among the two input matrices.
Then, we
combine the idea of multiplication of rectangular submatrices of the input
Boolean matrices with that of using
the quantum minimum search of D\"{u}rr and H{\o}yer
in order to obtain our fastest quantum
algorithm for MW running in $\tilde{O}(n^{2+\lambda/2})$ time, where
$\lambda $ satisfies the equation $\omega(1, \lambda, 1) = 1 + 1.5 \,
\lambda $.
By the currently best bounds on
$\omega(1,\lambda,1)$, the
\vfill
\newpage
\noindent
running time of our algorithm is
$O(n^{2.434})$
\footnote{In the upper bound on the 
running time of our fastest quantum algorithm
for MW,
we could replace $\omega(\ )$ by its generalization
to include the model of quantum computation.
However, since no quantum algorithms
for Boolean matrix product faster than those algebraic 
ones in the general case are known so far, this would not yield
an improvement at present.}.
We obtain the same asymptotic upper time-bounds
for the aforementioned problems related to MW. 

In the second part of our paper, we consider a relaxed version of the MW
problem (in the standard model) asking
for reporting a witness of bounded rank (the maximum  witness
has rank 1) for each non-zero entry of the matrix product.
First, by adapting the fastest known algorithm for maximum
witnesses, we obtain an algorithm for the relaxed problem
that reports for each non-zero entry of the product matrix
a witness of rank at most $\ell$ in time
$\tilde{O}((n/\ell)n^{\omega(1,\log_n \ell,1)}).$
  Then, by reducing the relaxed problem to the so called $k$-witness
problem, we provide an algorithm that reports for each
non-zero entry $C[i,j]$ of the product matrix $C$
a witness of rank $O(\lceil W_C(i,j)/k\rceil )$ with
high probability, where $W_C(i,j)$
is the number of witnesses for $C[i,j]$.
The algorithm runs in $\tilde{O}(n^{\omega}k^{0.4653} +n^{2+o(1)}k)$
time,
where
$\omega$ is the exponent of fast $n\times n$ matrix multiplication.

\subsection{Organization}
In Preliminaries, we provide some basic
notions and/or facts on matrix multiplication
and quantum computation. In Section 3, we present
our basic procedure for searching an interval of
indices for the maximum witness, the straightforward
quantum algorithm for MW implied by the procedure, and the
output-sensitive and input-sensitive refinements
of the algorithm. In Section 4, we present
and analyze our fastest in the general case
quantum algorithm for MW.
In Section 5, we
provide applications of our quantum algorithms
to the problems related to MW.
In Section 6, we present our approximation algorithms
for MW in the standard computational model.
We conclude with final remarks.
\junk{For applications of our quantum algorithms
to the problems related to MW as well as
additional remarks the reader is referred to
the full version \cite{KL}.
Let $\omega(1,r,1)$ denote the exponent of fast arithmetic
multiplication of an
$n \times n^r$ matrix by an $n^r \times n$ matrix.
In particular, $\omega(1,1,1)$ denoted by $\omega$
is known to not exceed $2.373$
\cite{LG14,Vassilevska12}. Next, let 
the notation $\tilde{O}(\ )$ suppress
poly logarithmic in $n$ factors.
Our main contributions are as follows:
}
\section{Preliminaries}

For a positive integer $r,$ we shall denote
the set of positive integers not greater than $r$
by $[r].$

For a matrix $D,$ $D^t$ denotes its transpose.

A \textit{witness} for a non-zero entry $C[i, j]$ of the Boolean matrix product
$C$ of a Boolean $p\times q$ matrix  $A$ and  a Boolean $q\times r$ matrix $B$ is any index $k\in [q]$ such that $A[i,k]$
and $B[k, j]$ are equal to 1. The number of witnesses for $C[i,j]$
is denoted by $W_C(i,j).$
A witness $k$ for $C[i, j]$ is of rank $h$ if there are exactly $h-1$
witnesses for this entry greater than $k.$ The witness of rank $1$
is the {\em maximum witness} for $C[i, j]$.
The \textit{witness problem} is to report
a witness for each non-zero entry of the Boolean matrix product
of the two input matrices.
The \textit{maximum witness problem} (MW) is to report
the maximum witness for each non-zero entry of the Boolean matrix product
of the two input matrices.

Recall that for natural numbers $p,\ q,\ r,$ $\omega(p,q,r)$ denotes the exponent 
of fast matrix multiplication for rectangular matrices
$n^p\times n^q$ and $n^q\times n^r,$ respectively.
For convenience, $\omega =\omega(1,1,1).$
The following recent upper bound on $\omega $ is
due Alman and Vassilevska Williams \cite{AV21}.
\begin{fact}
\label{fact: omega}
The fast matrix multiplication algorithm for
$n\times n$ matrices runs in $O(n^{\omega})$
time, where $\omega$ is not greater than $2.37286$
\cite{AV21} (cf. \cite{LG14,Vassilevska12}).
\end{fact}

Alon and Naor provided almost equally fast solution
to the witness problem for square Boolean matrices in \cite{AN96}.
It can be easily generalized to include the Boolean product
of two rectangular Boolean matrices of sizes $n\times n^q$
and $n^q\times n,$ respectively.
The asymptotic matrix multiplication time $n^{\omega}$
is replaced by $n^{\omega(1,q,1)}$ in the generalization.

\begin{fact}
\label{fact: wit}
For $q\in (0,1],$
the witness problem for the Boolean matrix
product of an $n\times n^q$ Boolean matrix
with an $n^q\times n$ Boolean matrix
can be solved (deterministically) in $\tilde{O}(n^{\omega(1,q,1)})$
time.
\end{fact}

Let $\alpha$ stand
for  $sup \{0
\le r \le 1 : \omega(1,r,1) = 2 + o(1) \}.$
The following recent lower bound on
$\alpha$ is due to Le Gall and Urrutia \cite{LGU}.

\begin{fact}\label{fact: ur}
The inequality $\alpha >  0.31389$ holds \cite{LGU}.
\end{fact}

Coppersmith \cite{C97} and Huang and Pan \cite{HP} proved the
following fact.

\begin{fact}
\label{fact:1}
The inequality 
$\omega(1,r,1) \le \beta(r)$ holds, where $\beta (r) = 2 + o(1)$ for $r
\in [0, \alpha]$ and $\beta(r) = 2 + \frac{\omega - 2}{1 - \alpha}
(r - \alpha) + o(1)$ for $r \in [\alpha, 1]$ \cite{C97,HP}.
\end{fact}

It will be the most convenient
to formulate our quantum algorithms in the model of
Quantum Random Access Machine (QRAM) saving the reader
a lot of technical details of alternative
formulations in the quantum circuit model \cite{NC}.
Thus, our quantum algorithm can access any entry of
any input matrix $A$ in an access random manner (cf. \cite{LeGallisaac12,NV}). 
More precisely, following \cite{NV},
we assume that there is an oracle
$O_A$ which for $i,\ j \in [n]$ and $z\in \{ 0,\ 1\}^*,$
maps the state $|i\rangle|j\rangle|0\rangle|z\rangle$
into the state $|i\rangle|j\rangle|A[i,j]\rangle|z\rangle.$
When a whole table $T$ is stored in the random access
memory of QRAM such an  oracle $o_T$ corresponding to $T$
is implicit.  
We shall estimate the time complexity of our quantum
algorithms in the unit cost model, in particular we shall
assign unit cost to each call to an oracle.
In case the time complexity of our quantum algorithm
exceeds the size of the input
matrices, we may assume w.l.o.g. that the input
matrices are just read into the QRAM memory.

Following Le Gall \cite{LeGallisaac12}, we can generalize the definition
of a quantum algorithm for Boolean matrix product
to include the MW problem.

\begin{definition}
A quantum algorithm for witnesses of Boolean matrix product (or the MW problem)
is a quantum algorithm that when given access to oracles $O_A$ and
$O_B$ corresponding to Boolean matrices $A$ and $B$, computes
with probability at least $2/3$ all the non-zero entries of the
product $A\times B$ along with one
witness (the maximum witness, respectively) for each non-zero entry.
\end{definition}

Note the probability of at least $\frac 23$ can be enhanced to at least
$1-n^{-\gamma}$, for $\gamma \ge 1,$ by iterating the algorithm 
$O(\log n)$ times. When the size of the input is bounded by $poly(n)$,
one uses the term {\em almost certainly} for the latter probability.

In fact, all our quantum algorithms for MW but the output
sensitive one report also ``No'' for each zero entry
of the product matrix.

\section{Quantum search for the  maximum witness}

One can find the maximum witness for
a given entry of the Boolean product
of two Boolean $n\times n$ matrices
in $\tilde{O}(\sqrt n)$ time with high probability
by recursively using Grover's quantum search \cite{Gr}
interleaved with a binary search.
However, the most convenient is to
use a specialized variant
of Grover's search due to D\"{u}rr and H{\o}yer \cite{A04,DH96}
for finding an entry of the minimum value in a table.
\begin{fact}\label{fact:dur}(D\"{u}rr and H{\o}yer \cite{DH96})
Let $T[k],$ $1\le k\le n$ be an unsorted table where
all values are distinct. Given an oracle for $T,$
the index $k$ for which $T[k]$ is minimum can be 
found by a quantum algorithm with probability
at least $\frac 12$ in $O(\sqrt n)$ time.
\end{fact}

Using this fact, we can design the following
procedure $MaxWit(A,B,i,j)$
returning the maximum witness  of the entry
$C[i,j]$ (if any) of the product
$C$ of two Boolean $n\times n$ matrices
$A$ and $B.$
\par
\vskip 4pt
\noindent
{\bf procedure} $MaxWit(A,B,i,j)$
\par
\noindent
{\em Input:} oracles corresponding to
a Boolean $p\times q$ matrix $A$
and a Boolean $q \times r$ matrix $B$,
and indices $i \in [p],\ j \in [r].$
\par
\noindent
{\em Output:} if the $C[i,j]$ entry of
the Boolean product $C$ of $A$ and $B$ 
has a witness then its  maximum
witness in $[q]$ otherwise ``No''.
\begin{enumerate}
\item $n\leftarrow \max \{ p,q,r\}$
\item
Define an oracle for a virtual, one-dimensional
integer table 
$T[k],$ $k\in [q]$ by
$T[k]=2n- A[i,k]B[k,j]n -k$.
\item
Iterate $O(\log n)$ times
the algorithm of D\"{u}rr and H{\o}yer 
for $T$ 
and set $k'$ to the index 
minimizing $T.$
\item
If $T[k']< n$ then return $k'$
as the maximum witness otherwise
return ``No''.
\end{enumerate}

\begin{lemma} \label{lem: grover}
Let $\beta$ be a positive integer.
By repetitively using
the algorithm of D\"{u}rr and H{\o}yer,
$MaxWit(A,B,i,j)$ can be implemented
in $\tilde{O}(\beta \sqrt {n})$ time
such that it returns a correct answer
with probability at least $1-n^{-\beta}$.
\end{lemma}
\begin{proof}
To begin with observe that for
$k,\ k' \in [q],$ if $k \neq k'$ 
then $T[k]\neq T[k'].$
This obviously holds for $k,\  k' \in [q]$ if
$A[i, k]B[k,j]=A[i, k']B[k',j]$ 
as well when $A[i, k]B[k,j]\neq A[i, k']B[k',j]$.
Furthermore, the value of $T[k]$ can
be computed with the help of the oracles
for $A$, $B$ in constant time 
in the QRAM model.
Next, suppose that the minimum value of $T$
is achieved for the index $k'$.
It is easily seen that if $T[k']<n$ then $k'$
is the maximum witness of $C[i,j]$ 
and otherwise $C[i,j]$ does not have
any witness.
By running the minimum search algorithm
of D\"{u}rr and H{\o}yer
$O(\beta \log n )$ times,
we can identify the maximum witness
of $C[i,j]$ with 
probability at least $1-n^{-\beta}$
in $\tilde{O}(\beta \sqrt {n})$
time.
\qed
\end{proof}
\subsection{A straightforward quantum algorithm for MW}

By Lemma \ref{lem: grover}, a straightforward $\tilde{O}(n^{2.5})$-time 
method for  MW is just to run the procedure
 $MaxWit(A,B,i,j)$ with appropriately large
constant $\beta$ for each
entry $C[i,j]$ of  the product matrix $C.$
See Algorithm 1 for a pseudo-code of this method.
\par
\vskip 3pt
\noindent
{\bf Algorithm 1}
\par
\noindent
{\em Input:} 
oracles corresponding to Boolean $n\times n$ matrices $A,\ B.$
\par
\noindent
{\em Output:} maximum witnesses for all non-zero
entries of the Boolean product of $A$ and $B$.
and ``No''  for all zero entries of the product.
\par
\noindent
$\ \ \ \ $ {\bf for} all $i,\ j\in [n]$ {\bf do} 
\par
\noindent
$\ \ \ \ $ $MaxWit(A,B,i,j)$ 
\par
\vskip 3pt
\noindent
Note that Algorithm 1 returns also ``No'' 
for zero entries of $C.$

By 
Lemma \ref{lem: grover}
with sufficiently large $\beta$, 
we obtain immediately the following theorem.
\begin{theorem}
Algorithm 1 solves the MW problem 
in $\tilde{O}(n^{2.5})$ time.
\end{theorem}

\subsection{An output-sensitive quantum algorithm for MW}

By adding as a preprocessing
a known output-sensitive quantum algorithm
for the Boolean product of the matrices $A$ and $B,$
Algorithm 1 can be transformed
into an output-sensitive one.
\par
\vskip 4pt
\noindent
{\bf Algorithm 2}
\par
\noindent
{\em Input:} 
oracles corresponding to Boolean $n\times n$ matrices $A,\ B.$
\par
\noindent
{\em Output:} maximum witnesses for all non-zero
entries of the Boolean product of $A$ and $B$.
\noindent
\begin{enumerate}
\item
Run an output-sensitive
quantum algorithm for the Boolean product $C$ of $A$ and $B.$
\item
{\bf for} all non-zero entries $C[i,j]$ {\bf do}
\newline
$MaxWit(A,B,i,j)$
\end{enumerate}
\junk{
\par
\noindent
{\bf for} all $i,\ j\in [n]$ {\bf do} 
\newline
$MaxWit(A,B,i,j)$
\par
\vskip 4pt
\noindent
By 
Lemma \ref{lem: grover}
with sufficiently large $\beta$, 
we obtain immediately the following theorem.
\begin{theorem}
Algorithm 1 solves the MW problem 
in $\tilde{O}(n^{2.5})$ time.
\end{theorem}}
\begin{theorem}
The MW problem
can be solved by a quantum 
algorithm in $\tilde{O}(n\sqrt s+ s\sqrt n)$ time, where $s$ is the number
of non-zero  entries in the product.
 \end{theorem}
\begin{proof}
Consider Algorithm 2.
Due to Step 1, the procedure $MaxWit$ is called
only for non-zero entries of $C$.
Hence, the total time
taken by Step 2 is $\tilde{O}(s\sqrt n)$ by Lemma \ref{lem: grover}
with any fixed $\beta$.
It is sufficient now to plug in the output-sensitive
quantum algorithm for Boolean matrix
product due to Le Gall
\cite{LeGallisaac12} running in $\tilde{O}(n\sqrt s + s\sqrt n)$ time
to implement Step 1. In order to obtain enough large probability
of the correctness of the whole output, we can iterate the plug in algorithm
a logarithmic number of times and pick  enough large $\beta$
in Lemma \ref{lem: grover}.
We obtain the
output-sensitive upper bound claimed in the theorem.
\qed
\end{proof}

Interestingly enough, the asymptotic time complexity of
our output-sensitive quantum algorithm for MW coincides with 
that of the output-sensitive quantum algorithm for Boolean matrix product
due Le Gall \cite{LeGallisaac12,LeGallsoda12}.

\subsection{An input-sensitive quantum algorithm for MW}

We can also refine the straightforward quantum
algorithm for MW in order to obtain
an input-sensitive quantum algorithm for MW.
\par
\vskip 4pt
\noindent
{\bf Algorithm 3}
\par
\noindent
{\em Input:} Boolean $n\times n$ matrices $A,\ B.$
\par
\noindent
{\em Output:} maximum witnesses for all non-zero
entries of the Boolean product of $A$ and $B$
and ``No''  for all zero entries of the product.
\par
\noindent
\begin{enumerate}
\item
For each column $j$ of the matrix $B$ 
compute the sequence $K_j$ of indices $k \in [n]$ 
in decreasing order such that
$B[k,j]=1$ by using the oracle for the matrix $B.$
Construct a one dimensional integer table $S_j$
of length $|K_j|$ such that for $s \in [|K_j|]$,
$S_j[s]$ is the $s$-th largest
element in $K_j$.
\item
{\bf for} all $i,\ j\in [n]$ {\bf do}
\begin{enumerate}
\item
Define an oracle for a virtual,
 one-dimensional integer table $T_{i,j}$ of length $|K_j|$
such that for $s \in [|K_j|]$,
$T_{i,j}[s]=2n-A[i, S_j[s]]n-S_j[s]$.
(The value of $T_{i,j}[s]$ can be retrieved
in constant time by using the oracle for the matrix $A$
and the table $S_j.$)
\item
Iterate $O(\log n)$ times
the algorithm of D\"{u}rr and H{\o}yer 
for $T_{i,j}$ 
and set $s'$ to the index 
minimizing $T_{i,j}.$
\item
If $T_{i,j}[s']< n$ then return $S_j[s']$ (i.e.,$n-T_{i,j}[s']$)
as the maximum witness for $C[i,j]$ otherwise
return ``No'' for $C[i,j]$.
\end{enumerate}
\end{enumerate}

An analysis of Algorithm 1 yields the following theorem.

\begin{theorem}
The MW problem for the Boolean product
of two Boolean
 $n\times n$ matrices,
with $m_1$ and $m_2$ non-zero
entries respectively, admits
a quantum algorithm running
in $\tilde{O}(n^2+n^{1.5}\sqrt {\min \{ m_1,m_2\}})$ time.
\end{theorem}
\begin{proof}
Consider Algorithm 3. Its correctness follows from
the definition of the tables
$T_{i,j}$, in particular the fact that each of them  has distinct values.
Let us estimate the time complexity
of Algorithm 3. We may assume w.l.o.g. that the number
of non-zero entries in the matrix $B$ is $m_2.$
Steps 1, 2(a) and 2(c) can be easily done in $\tilde{O}(n^2)$ total time.
In Step 2(b),
computing the maximum witnesses for
the entries in the $i$-th row of the 
product matrix takes 
$\tilde{O}(\sum_{j=1}^n\sqrt {|K_j|})$ time by Fact \ref{fact:dur}.
Since $\sum_{j=1}^n|K_j|\le m_2$ and the arithmetic mean
does not exceed the quadratic one,
we obtain
$\sum_{j=1}^n\sqrt {|K_j|}\le n\sqrt {\frac {m_2}n}.$
Consequently, Algorithm 3 runs in
$\tilde{O}(n^2+n^2\sqrt {\frac {m_2}n})$ time.

As in case of Algorithms 1 and 2, we can pick 
enough large constant at $\log n$ in
the upper bound on the number of iterations
of the algorithm of D\"{u}rr and H{\o}yer in order
to guarantee that the whole output of Algorithm 3 is
correct with probability at least $\frac 23 $.
Hence, by the time analysis of
Algorithm 3 and $A\times B=(B^t\times A^t)^t,$
we obtain the theorem.
\qed
\end{proof}

\section{The fastest method: combining rectangular Boolean matrix multiplication with 
quantum search}
 \begin{figure}
\begin{center}
\includegraphics[scale=0.2 ]{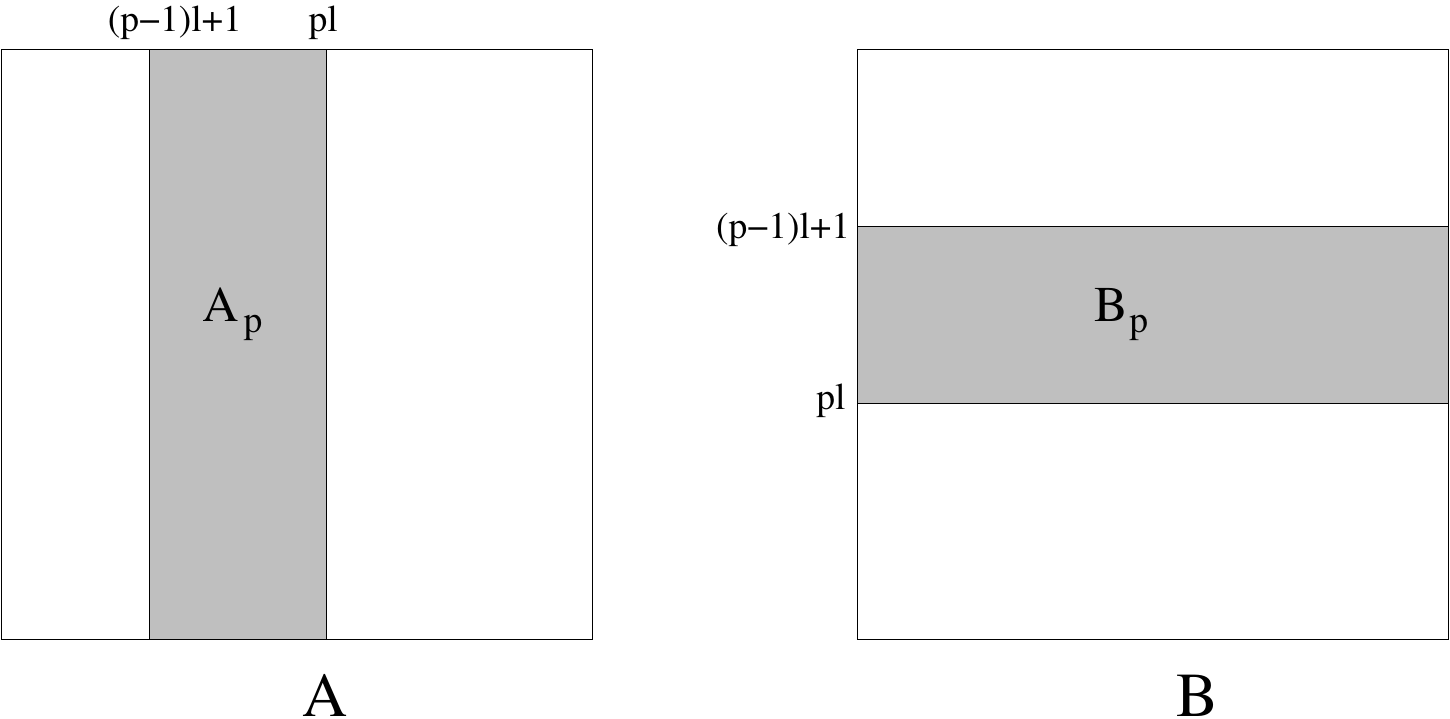}
\end{center}
\caption{The input matrices $A$ and $B$ are divided into 
vertical and horizontal strip submatrices $A_p$ and $B_p,$
respectively, in Algorithm 4.}
\end{figure}
The best known algorithm for MW from \cite{CK07}
relies on  the multiplication of rectangular submatrices of the input
matrices. We can combine this idea
 with that of our procedure $MaxWit$ based on 
the quantum search for the minimum in order to obtain our fastest quantum
algorithm for MW.
\vfill
\newpage
\noindent
{\bf Algorithm 4}
\par
\noindent
{\em Input:}
oracles
corresponding to Boolean $n\times n$ matrices $A,\ B,$ and
a parameter $\ell \in [n].$
\par
\noindent
{\em Output:} maximum witnesses for all non-zero
entries of the Boolean product of $A$ and $B,$
and ``No''  for all zero entries of the product.
\par
\noindent
\begin{enumerate}
\item Divide $A$ into $\lceil n/\ell \rceil$ vertical strip
submatrices $A_1,...,A_{\lceil n/\ell \rceil}$ of width $\ell$
with the exception of the last one that can have width $\le \ell .$
\item Divide $B$ into $\lceil n/\ell\rceil$ horizontal strip
submatrices $B_1,...,B_{\lceil n/\ell \rceil}$ of width $\ell$
with the exception of the last one that can have width $\le \ell.$
\item {\bf for} $p\in [\lceil n/\ell \rceil ]$ compute the Boolean product $C_p$
 of $A_p$ and $B_p$
\item {\bf for} all $i,\ j\in [n]$ {\bf do}
\begin{enumerate}
\item Find the largest $p$ such that $C_p[i,j]=1$ or set $p=0$ if it does not exist.
\item {\bf if} $p>0$ {\bf then} return $\ell(p-1)+
MaxWit(A_p,B_p,i,j)$ {\bf else} return ``No''
\end{enumerate}
\end{enumerate}
\begin{figure}
[h]
\begin{center}
\includegraphics[scale=0.32]{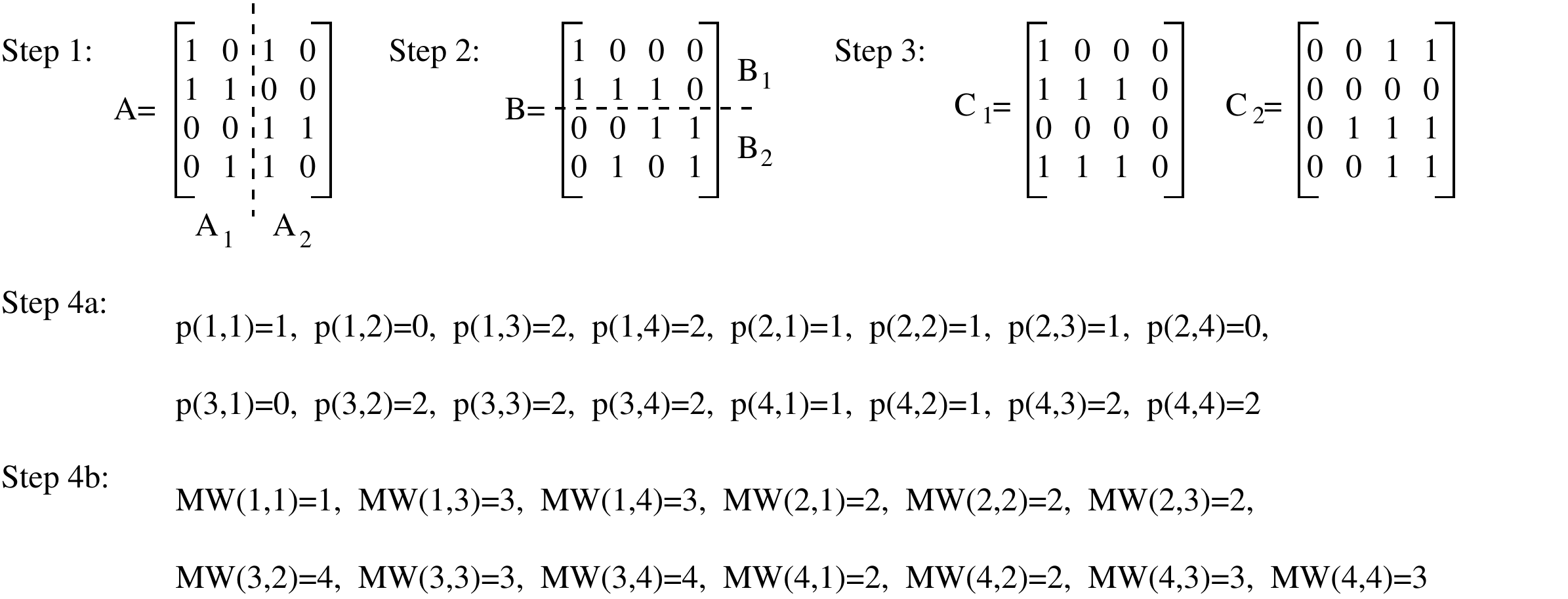}
\end{center}
\caption{An example illustrating Algorithm 4.}
\end{figure}

\begin{lemma} \label{lem: timealg4}
Algorithm 4 runs in time
$\tilde{O}((n/\ell)n^{\omega(1,\log_n\ell,1)}+n^3/\ell + n^2\sqrt \ell)$.
\end{lemma}
\begin{proof}
Steps 1, 2, take $O(n^2)$ time.
Step 3 requires $O((n/\ell)n^{\omega(1,\log_n \ell,1)})$ time.
Step 4(a) takes $O(n^2\times n/\ell)$ time totally.
Finally, Step 4(b) requires $\tilde{O}(n^2\sqrt \ell)$ time
totally by Lemma \ref{lem: grover}.
\qed
\end{proof}
By Lemma \ref{lem: grover} with sufficiently
large $\beta$
and the time analysis in Lemma \ref{lem: timealg4},
we obtain the following trade-offs between
preprocessing time and answering a maximum witness query
time depending on $\ell .$
\junk{we obtain the following trade-offs between 
preprocessing and answering a maximum witness query.}
\begin{theorem}\label{theo: tradeoffs}
Let $C$ denote the Boolean product
of two Boolean $n\times n$ matrices,
and let $i,\ j$ be any integers in $[n].$
Without any preprocessing,
a maximum witness query for $C[i,j]$
can be answered in $\tilde{O}(\sqrt n )$ time.
Let $\ell$ be a parameter  in $[n].$
After an $O((n/\ell)n^{\omega(1,\log_n \ell,1)})$
time preprocessing (Steps 1,2,3 in Algorithm 4),
a maximum witness query for $C[i,j]$
can be answered in $\tilde{O}(n/\ell+ \sqrt \ell)$
time. After an 
$O((n/\ell)n^{\omega(1,\log_n \ell,1)}+n^{3}/\ell)$
time preprocessing (Steps 1,2,3, 4(a) in Algorithm 4),
a maximum witness query for $C[i,j]$
can be answered in $\tilde{O}(\sqrt \ell )$ time.
Finally, after running the whole Algorithm 4
 in time
$\tilde{O}((n/\ell)n^{\omega(1,\log_n \ell,1)}+n^3/\ell + n^2\sqrt \ell)$,
a maximum witness query for $C[i,j]$
can be answered in $O(1)$ time.
\end{theorem}
\subsection{Finding $\ell$ minimizing the total time.}
Recall 
that $\omega(1,r,1)$ denotes the exponent of the multiplication of an
$n \times n^r$ matrix by an $n^r \times n$ matrix.
By Lemma \ref{lem: timealg4},
the total time taken by Algorithm 4 for maximum witnesses is
\begin{displaymath}
   \tilde{O}((n/\ell) \cdot n^{\omega (1, \log_n\ell, 1)}+ n^3/\ell + n^2 \, \sqrt{\ell}).
\enspace.
\end{displaymath}
By setting $r$ to $\log_n\ell$ our upper bound transforms to
$\tilde{O}(n^{1-r+ \omega (1, r , 1)}+ n^{3-r} + n^{2+r/2})$. Note that by
assuming $r \ge \frac23$, we can get rid of the additive $n^{3-r}$
term. Hence, by solving the equation $1 - \lambda + \omega (1,
\lambda, 1) = 2 + \lambda/2$ implying $\lambda \ge \frac23$ by
$\omega(1, \lambda, 1) \ge 2$ and setting
sufficiently large $\beta$ in Lemma \ref{lem: grover} , we obtain our main result.

\begin{theorem}\label{theo: wit}
Let $\lambda$ be such that $\omega (1, \lambda, 1) = 1 + 1.5 \,
\lambda $. The maximum witnesses for all non-zero entries of the
Boolean product of two $n \times n$ Boolean matrices can be computed
almost certainly
by a quantum algorithm in $\tilde{O}(n^{2 + \lambda/2 })$ time.
\end{theorem}
\junk{
Coppersmith \cite{C97} and Huang and Pan \cite{HP} proved the
following fact,
 where the bounds on $\omega$ and $\alpha$
are updated by the more recent results from \cite{LG14} and \cite{LGU}, respectively.

\begin{fact}{\bf \cite{C97,HP}}
\label{fact:1}
Let $\omega = \omega (1,1,1) < 2.3728639$ and let $\alpha = sup \{0
\le r \le 1 : \omega(1,r,1) = 2 + o(1) \} > 0.31389.$
Then
$\omega(1,r,1) \le \beta(r)$, where $\beta (r) = 2 + o(1)$ for $r
\in [0, \alpha]$ and $\beta(r) = 2 + \frac{\omega - 2}{1 - \alpha}
(r - \alpha) + o(1)$ for $r \in [\alpha, 1]$.
\end{fact}}

Note that by Fact \ref{fact:1}, the solution $\lambda$ of the
equation $\omega (1, \lambda, 1) = 1 + 1.5 \, \lambda $ is satisfied
by 
$\lambda = \frac{1 - \alpha \, (\omega - 1)}{1.5 \, (1 - \alpha)
- (\omega - 2)} + o(1)$. 
Note also that $\lambda$ is increasing in $\omega$ 
and decreasing in $\alpha .$
Hence, the inequality $\lambda < 0.8671$ holds 
by Fact \ref{fact: omega} and Fact \ref{fact: ur}.
We obtain the
following concrete corollary.

\begin{corollary}\label{cor: wit}
The 
maximum witnesses for all non-zero entries of the Boolean
product of two $n\times n$ Boolean matrices can be computed
almost certainly
by a quantum algorithm in $\tilde{O}(n^{2.4335})$ time.
\end{corollary}
\junk{
By combining Theorem \ref{theo: max0} with Theorem \ref{theo: wit} and
Corollary \ref{cor: wit}, we obtain also an $O(n^{2.575})$-time
solution to the all-pairs LCA problem in dags.

\begin{theorem}\label{theo: lambda}
Let $\omega (1, \lambda, 1) = 1 + 2 \, \lambda$. The all-pairs LCA
problem for an arbitrary dag with $n$ vertices can be solved in
time $O(n^{2 + \lambda })$, which is bounded above by $O(n^{2.575})$.
\end{theorem}
}

\section{Applications of quantum algorithms for MW}

The problem of finding a \emph{lowest common ancestor} (LCA) in a
tree, or more generally, in a \emph{directed acyclic graph} (dag) is
a basic problem in algorithmic graph theory.
A LCA of vertices $u$ and
$v$ in a dag is an ancestor of both $u$ and $v$ that has no
descendant which is an ancestor of $u$ and $v$, see Fig. 3.
We consider the problem of preprocessing a dag such that
LCA queries can be answered quickly for any pair of vertices.
The all-pairs LCA problem is to compute LCA for all pairs
of vertices in the input dag.
\begin{figure}
\begin{center}
  \includegraphics[scale=0.4]{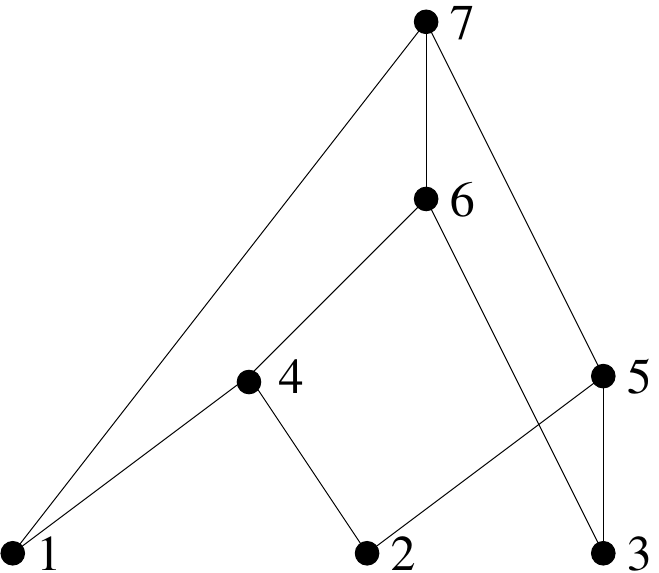}
\end{center}
\caption{An example of a dag. 
Note that the vertices 5 and  6 are LCA for the vertices  2,\ 3,
and the vertex 6 is also an  LCA for the vertices 1,\ 3
but it is not an LCA for the vertices  1,\ 2.}
\end{figure}
In the proof of Theorem 11 in \cite{CK07}, on the basis of an input
$n$-vertex dag, a Boolean $n\times n$ matrix $A$ is constructed in
$O(n^{\omega})$ time such that the maximum witness for $C[i,j]$, where
$C=A\times A^t$ yields an LCA for vertices $i,\ j$ in the dag.  Combining
this with Corollary \ref{cor: wit}, we obtain 
also the following theorem.
\junk{Theorem \ref{theo: tradeoffs}, we obtain the following
  trade-offs between preprocessing time and query time for LCA queries
  in the quantum computational model.

\begin{theorem}\label{theo: lcatradeoffs}
Let $D$ be a dag
on $n$ vertices and let
$\ell$ be a parameter in $[n]$.
After $O(n^{\omega})$-time preprocessing,
any LCA query in $D$
can be answered in $\tilde{O}(\sqrt n )$ time.
After an $O((n/\ell)n^{\omega(1,\log_n \ell,1)})$
time preprocessing,
any LCA query in $D$
can be answered in $\tilde{O}(n/\ell+ \sqrt \ell)$
time. After an $O((n/\ell)n^{\omega(1,\log_n \ell,1)}+n^{3}/\ell))$
time preprocessing,
any LCA query in  $D$
can be answered in $\tilde{O}(\sqrt \ell )$ time.
Finally, after preprocessing
 in time
$\tilde{O}((n/\ell)n^{\omega(1,\log_n \ell,1)}+n³/\ell + n^2\sqrt \ell)$,
any LCA query in $D$
can be answered in $O(1)$ time.
\end{theorem}

By Theorem 11 in \cite{CK07}
and Corollary \ref{cor: wit}, we obtain 
also the following theorem.}

\begin{theorem}
The all-pairs LCA problem can be solved
by a quantum algorithm in $O(n^{2.4335})$ time.
\end{theorem}
Very recently, Grandoni et al. have presented an $\tilde{O}(n^{2.447})$-time algorithm for
the LCA problem in the standard computational model \cite{GIL}.

Shapiro et al. considered the following 
all-pairs bottleneck weight path problem
 in directed, vertex weighted graphs in \cite{SYZ11}.
Let $G=(V,E)$ be a directed, vertex-weighted
graph. The {\em bottleneck} weight of a directed path
in $G$ is the minimum weight of a vertex on the path. 
For two vertices $u,\ v$ of $G,$ the bottleneck
weight from $u$ to $v$ is the maximum bottleneck
weight of a directed path from $u$ to $v$ in $G.$
The all-pairs bottleneck paths problem (APBP) 
is to find bottleneck weights for all ordered
pairs of vertices in $G.$ The authors of \cite{SYZ11}
considered two variants of APBP,
an open variant where the weights of 
the start and end vertices are not counted
and a closed variant where 
the weights of the start
and end vertices are counted. In particular,
in Theorem 2 in \cite{SYZ11}, they show that
both variants of APBP,
MW, and the problem
of computing maximum weight of two-edge paths 
between all pairs of vertices in vertex weighted
graphs are computationally equivalent 
(up to constant factors).
Hence, by Corollary \ref{cor: wit}, we obtain the following theorem.

\begin{theorem}
The following problems admit
an $\tilde{O}(n^{2.434})$-time quantum
algorithm: Open APBP, Closed APBP,
the all-pairs maximum weight two-edge  paths
in vertex weighted graphs.
\end{theorem}

As a corollary, we obtain a faster quantum
algorithm for the problem considered in
\cite{VWY10}.

\begin{corollary} 
Let $G$ be an undirected
vertex-weighted graph on $n$
vertices.
The problem of finding  for each  edge $\{u,v\}$ of $G,$ a
heaviest  (or, lightest) triangle $\{u,v,w\}$ in $G$
admits a quantum algorithm running 
 in $O(n^{2.434})$ time. 
\end{corollary}

\section{Approximation algorithms}
In this section, we present two approximation approaches to
MW in a standard computational model.
The first approach follows the idea of the fastest known algorithm
for MW \cite{CK07} but instead of searching the final index intervals
where the respective maximum witnesses are localized 
some witnesses from the intervals are reported. The second approach relies
on the repetitive applying the deterministic algorithm
for multiple witnesses from \cite{GKL08} and  the goodness
of its approximation for a matrix product entry depends on the number of witnesses
for the entry.

\subsection{The method based on rectangular matrix multiplication}

By slightly modifying the algorithm for MW \cite{CK07}
(or, the quantum Algorithm 4)
based on fast rectangular multiplication, we can obtain a faster
approximation algorithm. For a given $\ell $,
it reports for each non-zero entry of the Boolean
matrix product a witness of rank not exceeding $\ell$
instead of the maximum witness. In the time analysis of the approximation
algorithm, we rely on
the fact that witnesses for non-zero
entries of the Boolean product of two Boolean matrices can be reported
in time proportional to the time taken by fast Boolean matrix multiplication
up to polylogarithmic factors (see Fact \ref{fact: wit}).
\par
\vskip 4pt
\noindent
{\bf Algorithm 5}
\par
\noindent
{\em Input:}
Boolean $n\times n$ matrices $A,\ B,$ and
a parameter $\ell \in [n].$
\par
\noindent
{\em Output:} witnesses for all non-zero
entries of the Boolean product of $A$ and $B$ having rank not exceeding $\ell$
and ``No''  for all zero entries of the product.
\noindent
\begin{enumerate}
\item Divide $A$ into $\lceil n/\ell \rceil$ vertical strip
submatrices $A_1,...,A_{\lceil n/\ell \rceil}$ of width $\ell$
with the exception of the last one that can have width $\le \ell .$
\item Divide $B$ into $\lceil n/\ell\rceil$ horizontal strip
submatrices $B_1,...,B_{\lceil n/\ell \rceil}$ of width $\ell$
with the exception of the last one that can have width $\le \ell.$
\item {\bf for} $p\in [\lceil n/\ell \rceil ]$ {\bf do}
 \newline
 Compute the Boolean product $C_p$
 of $A_p$ and $B_p$ along with single witnesses for all positive entries of the product
\item {\bf for} all $i,\ j\in [n]$ {\bf do}
\begin{enumerate}
\item Find the largest $p$ such that $C_p[i,j]=1$ or set $p=0$ if it does not exist.
\item {\bf if} $p>0$ {\bf then} return the found witness of  $C_p[i,j]$ {\bf else} return ``No''
\end{enumerate}
\end{enumerate}

\begin{lemma} \label{lem: app1}
Algorithm 5 runs in time
$\tilde{O}((n/\ell)n^{\omega(1,\log_n\ell,1)})$.
\end{lemma}
\begin{proof}
Steps 1, 2, take $O(n^2)$ time.
Step 3 requires $\tilde{O}((n/\ell)n^{\omega(1,\log_n \ell,1)})$ time
by a straightforward generalization of the $\tilde{O}(n^{\omega})$-time
algorithmic solution to the witness problem
for square Boolean matrices given in Fact \ref{fact: wit} to include
rectangular Boolean matrices.
Step 4(a) takes $O(n^2\times n/\ell)$ time totally.
Finally, Step 4(b) requires $O(n^2)$ time
totally. It remains to observe that the term
$\tilde{O}((n/\ell)n^{\omega(1,\log_n \ell,1)})$ dominates
the asymptotic time complexity of the algorithm
by $\omega(1,\log_n\ell,1)\ge 2.$
\qed
\end{proof}

\begin{theorem}
 For all non-zero entries 
of the Boolean matrix product of
two Boolean  $n\times n$ matrices,
witnesses of rank not exceeding $\ell$
can be reported
 in time $\tilde{O}((n/\ell)n^{\omega(1,\log_n\ell,1)})$.
\end{theorem}

\subsection{The method based on multi-witnesses}

A straightforward method to obtain single witnesses 
of rank $O(\lceil W_C(i,j)/k \rceil )$
for the nonzero entries $C[i,j]$ of the Boolean product 
$C$ of two
Boolean $n\times n$ matrices is to iterate
a randomized algorithm for single witnesses for
the entries of $C$ \cite{AN96}. After $O(k \log n)$
iterations such witnesses 
can be reported with high probability.
This straightforward method takes $\tilde{O}(n^{\omega}k)$
time. We provide a more efficient algorithm
for this problem based on the  algorithm 
for the so called $k$-witness problem from
\cite{GKL08}.

The {\em $k$-witness problem} for 
the Boolean matrix product of two $n\times
n$ Boolean matrices is to produce a list of 
$r$ witnesses for each positive entry
of the product, 
where $r$ is the minimum of $k$ and
the total number of witnesses for this entry.

In the following fact from \cite{GKL08}, the upper bounds have been updated
by incorporating the more recent results
on the parameters $\omega$ (Fact \ref{fact: omega}) and $\alpha$ \cite{LGU}. 

\begin{fact}\label{fact: kwit}\cite{GKL08}
There is a randomized algorithm solving
the $k$-witness problem almost certainly in time
$\tilde{O}(n^{2+o(1)}k + n^{\omega}k^{(3-\omega - \alpha)/(1-\alpha)})$,
where $\alpha \approx 0.31389$
(see Fact \ref{fact: ur}). 
One can rewrite
the upper time bound
as $\tilde{O}(n^{\omega} k^{\mu} + n^{2+o(1)} k)$, where $\mu\approx 0.46530$.
\end{fact}
\vskip 4pt
\noindent
{\bf Algorithm 6}
\par
\noindent
{\em Input:}  Boolean $n\times n$ matrices $A,\ B,$ and
a parameter $k \in [n]$ not less than $4.$
\par
\noindent
{\em Output:} single witnesses $Wit[i,j]$ for all non-zero
entries $C[i,j]$ of the Boolean product $C$ of $A$ and $B$
such that $rank(Wit[i,j])\le 4\lceil W_C(i,j)/k \rceil$ with probability
at least $\frac 12 -e^{-1}$.
\begin{enumerate}
\item $D\leftarrow B$
\item Initialize $n\times n$ integer matrix $Wit$ 
by setting all its entries to $0.$
\item {\bf for} $q=1,..., O(\log n)$  {\bf do}
\begin{enumerate}
\item Run an algorithm for the $k$-witness problem
for the product $F$ of the matrices $A$ and $D$.
\item For all $1 \le i,\ j\le n,$ set 
$Wit[i,j]$ to the maximum of $Wit[i,j]$ and
the maximum among the reported witnesses for $F[i,j]$.
\item Uniformly at random set each $1$ entry
of $D$ to zero with probability $\frac 12$.
\end{enumerate}
\end{enumerate}

$TW(n,k)$ will stand for the running
time of the k-witness 
algorithm 
for the Boolean product
of the two input Boolean matrices of size $n\times n$
used in Algorithm 6.

\begin{lemma}\label{lem: timerank}
Algorithm 6 runs in $\tilde{O}(TW(n,k)+n^2k)$ time.
\end{lemma}
\begin{proof}
The block of the while loop can be implemented
in $O(TW(n,k)+n^2k)$ time. It is sufficient to
observe that the block is iterated
$O(\log n)$ times.
\qed
\end{proof}

\begin{lemma}\label{lem: rank}
For $1\le i,\ j\le n$ 
and $k\ge 4,$ the final
value of $Wit[i,j]$ in Algorithm 5 is
a witness of $C[i.j]$ with rank
at most $4\lceil W_C(i,j)/k \rceil$ with probability
not less than $\frac 12 -e^{-1}$.
\end{lemma}
\begin{proof}
We may assume without loss of generality that
$W_C(i,j)/k> 1$ since otherwise the maximum
witness for $C[i,j]$ is found already
in the first iteration of the block of the while loop.
Let $\ell = \lceil \log_2 {W_C(i,j)/k }\rceil$.
A witness of the entry $C[i,j]$ survives
$\ell +1$ iterations of the block of the while
loop with probability $2^{-\ell -1}.$ Hence, 
after $\ell +1$ iterations of the block of the while
loop the expected number of witnesses of
the entry $C[i,j]$ that survive does not exceed $k/2$. 
Consequently, the number of witnesses of $C[i,j]$ that survive
does not exceed $k$ with probability at least $\frac 12.$
They are reported as witnesses of $F[i,j]$ in the $\ell +2$ iteration.
On the other hand, 
the probability 
that none of witnesses not greater
than $4W_C(i,j)/k $ survives the 
$\ell +1$ iterations is 
at most $(1-\frac 1{2^{\ell +1}})^{4W_C(i,j)/k}\le e^{-1}$
by $k\ge 4.$
Observe that for events $A$ and $B,$
$Prob(A \cap B)\ge 1-Prob(\bar{A} \cup \bar{B})\ge 1-Prob(\bar{A})
-Prob(\bar{B}).$
Hence, 
at least one witness of rank at most
$4W_C(i,j)/k $ survives $\ell +1$ iterations 
and it is reported in the $\ell +2$ iteration
with probability at least $1- \frac 12 -e^{-1}\ge \frac 12 -e^{-1}.$
\qed
\end{proof}

\begin{theorem}\label{theo: main2}
Let $C$ be the Boolean product
of two Boolean $n\times n$ matrices  
and let $k$ be an integer not less
than $4.$
One can compute for all non-zero
entries $C[i,j]$ single witnesses
of rank $O(\lceil W_C(i,j)/k\rceil )$
in $\tilde{O}(n^{\omega}k^{0.4653}+n^{2+o(1)}k)$
time almost certainly.
\end{theorem}
\begin{proof}
By Lemma \ref{lem: rank}, it is sufficient to
iterate Algorithm 5 $O(\log n)$ times
to achieve the probability of at
least $1-n^{-\beta}$, $\beta \ge 1.$
The time complexity bound follows
from Lemma \ref{lem: timerank}
by the upper bound on
$TW(n,k)$ from Fact \ref{fact: kwit}.
\qed
\end{proof}

By plugging the randomized upper bound on $CW(n,k)$ from Fact
\ref{fact: wit} into Theorem \ref{theo: main2} and assuming the
notation from the theorem, we obtain the following corollary.

\begin{corollary} \label{cor: approx}
There is a randomized
algorithm that for $4\le k\le n^{0.4212}$ computes for all non-zero
entries $C[i,j]$ single witnesses
of rank $O(\lceil W_C(i,j)/k\rceil )$ almost certainly
in time substantially subsuming
the best known upper time bound for 
computing maximum witnesses for all non-zero
entries of $C.$ In particular, if the number of
witnesses for each entry of $C$ is upper bounded
by $w \le n^{0.4212}$ then by setting $k=w,$
we obtain for all non-zero entries of $C$
a witness of rank $O(1)$ almost certainly,
substantially faster than maximum witnesses
for these entries.
\end{corollary}
\junk{
Importantly, 
Algorithm 5 and hence Theorem \ref {theo: main2}
and Corollary \ref{cor: approx} can be derandomized analogously as the algorithm
of Alon and Naor for the $1$-witness problem by 
using the so called small probability spaces keeping
the asymptotic upper bounds \cite{AN96}.}

\section{Final remarks}
Due to the quantum search for the minimum,
the MW problem is relatively easier in
the quantum computation model.  Already the straightforward quantum
algorithm (Algorithm 1) 
running in $\tilde{O}(n^{2.5})$ time, is substantially faster than the
best known algorithm for MW in the standard model running in
$O(n^{2.569})$ time (originally, $O(n^{2.575})$ time \cite{CK07}).
Also, the gap between our fastest algorithm for MW (Algorithm 4)
running in $O(n^{2.434})$ time and the fastest algorithm for Boolean
matrix product in the quantum computation model is substantially
smaller than the corresponding gap in the standard model. An
additional reason here is that no quantum algorithm for Boolean matrix
product in general case faster than the algebraic ones in the standard
model is known so far.  

Our input-sensitive quantum algorithm for MW (Algorithm 3)
is faster than our $O(n^{2.434})$-time algorithm for MW
(Algorithm 4) when one of the input $n\times n$ matrices has
a number of non-zero entries substantially smaller than $O(n^{1.868})$.
Similarly, our output-sensitive quantum algorithm for MW (Algorithm 2)
is faster than the $O(n^{2.434})$-time algorithm for MW
when the number $s$ of non-zero entries in
the product matrix is substantially smaller than  $O(n^{1.934})$.

Our approximation algorithm for MW could be used for example
to find triangles passing through specified edges 
approximating heaviest ones in vertex weighted graph (cf. \cite{VWY10}).

\section*{Acknowledgments}
The authors thank Francois Le Gall for a useful clarification
of the current status of quantum algorithms for
Boolean matrix product.
The research has been
supported in part by 
Swedish Research Council grant  621-2017-03750.

\end{document}